\documentclass[fleqn,10pt]{wlscirep}

\usepackage{graphicx}
\usepackage{amssymb}
\usepackage{amsmath,cases}
\usepackage{amsthm}
\usepackage{cite}
\usepackage{stfloats}
\usepackage{subfigure} 
\usepackage{epstopdf}
\usepackage{mdwlist}
\usepackage{threeparttable}
\usepackage[all,2cell]{xy} \UseAllTwocells
\usepackage{booktabs}
\usepackage{fancyvrb}
\usepackage{rotating}
\usepackage{rotfloat}
\usepackage{pifont}
\usepackage{verbatim}
\usepackage{color}
\usepackage{bm}

\usepackage{siunitx}
\usepackage{multirow}
\usepackage{array}
\usepackage{booktabs}
\usepackage{braket}

\newtheorem{theorem}{Theorem}
\newtheorem{lemma}{Lemma}

\newtheorem{proposition}{Proposition}

\newtheorem{remark}{Remark}
\newtheorem{property}{Property}

\newcommand{\Mod}[1]{\ \mathrm{mod}\ #1}

\DeclareRobustCommand{\BM}[1]{%
  \ifcat\noexpand#1\relax\bm{\boldUppercaseItalicGreek{#1}}\else\bm{#1}\fi
}
\newcommand{\dw}{\mathcal{N}_{\mathrm{dw}}}
\newcommand{\V}[1]{\bm{#1}}
\newcommand{\M}[1]{\BM{#1}}
\newcommand{\modA}[3]{\left(#1 + #2\right) \Mod{#3} }
\newcommand{\modS}[3]{\left(#1 - #2\right) \Mod{#3} }

\newcommand{\qcS}[1]{\mathcal{N}_{\mathrm{{#1}}}}
\newcommand{\qc}[2]{\qcS{#1}\left( #2\right)}
\newcommand{\ketbra}[2]{\ket{#1}\bra{#2}}

\definecolor{CCTLABgreen}{RGB}{0,128,0}

\title{Holevo Capacity of Discrete Weyl Channels}

\author[1]{Junaid~ur~Rehman}
\author[1,*]{Youngmin~Jeong}
\author[2]{Jeong~San~Kim}
\author[1,*]{Hyundong~Shin}
\affil[1]{Department of Electronic Engineering, Kyung Hee University, 1732 Deogyeong-daero, Giheung-gu, Yongin-si, Gyeonggi-do, 17104 Korea.}
\affil[2]{Department of Applied Mathematics and Institute of Natural Sciences, Kyung Hee University, 1732 Deogyeong-daero, Giheung-gu, Yongin-si, Gyeonggi-do, 17104 Korea.}

\affil[*]{Correspondence and requests for materials should be addressed to J.R and H.S (email: junaid@khu.ac.kr; hshin@khu.ac.kr)}

\begin{abstract}

Holevo capacity is the maximum rate at which a quantum channel can reliably transmit classical information without entanglement. However, calculating the Holevo capacity of arbitrary quantum channels is a nontrivial and computationally expensive task since it requires the numerical optimization over all possible input quantum states. In this paper, we consider discrete Weyl channels (DWCs) and exploit their symmetry properties to model DWC as a classical symmetric channel. We characterize lower and upper bounds on the Holevo capacity of DWCs using simple computational formulae. Then, we provide a sufficient and necessary condition where the upper and lower bounds coincide. The framework in this paper enables us to characterize the exact Holevo capacity for most of the known special cases of DWCs.

\end{abstract}
\begin{document}

\flushbottom

\maketitle

\thispagestyle{empty}

\section*{Introduction}

One of the fundamental tasks in the context of information theory is to compute the maximum rate at which information can be reliably transmitted \cite{WIL:11:CUP,CJ:12:JWS}. Classical channels have the capability of transmitting classical information only. On the contrary, quantum channels are more rich in terms of communication tasks\cite{RQJ:17:MPLB,QRJ:17:PTEP}. Trivially, quantum channels are capable of transmitting quantum information. However, due to the versatile nature and unique features of quantum mechanics, it is possible to associate multiple communication tasks with a quantum channel \cite{ZJS:18:SR}. Thus, we have classical capacity, quantum capacity, private classical capacity, and entanglement-assisted classical capacity of a quantum channel.  All of theses correspond to different information communication tasks \cite{Hol:98:T_IT,DEV:05:IEIT,BSS:02:T_IT,BSS:99:PRL}. 

The calculation of various capacities involves an optimization task that is not easy to perform. For example, the capacity of a classical channel is given by a single letter formula---the mutual information between input and output of the channel---maximized over the probability distribution of the input random variable \cite{Shan:48:BST}. Efficient methods exist that can perform this maximization \cite{Bla:72:T_IT,Ari:72:T_IT}. On the contrary, capacities (except the entanglement-assisted classical capacity) of a quantum channel are given in terms of regularization of asymptotically many channel uses. These regularized formulae are  mathematically intractable in general and put forth an unsolvable optimization problem \cite{CEM:15:NC}. Simplification of these formulae is not possible due to the nonadditive and nonconvex natures of capacities of quantum channels \cite{SY:08:Sci,ES:16:PRA,Has:09:NP}.
 The need of regularization, however, can be removed either 1) if the capacity of the channel is additive, or 2) if we restrict the optimization to be on the individual channel use.
For example, unital qubit channels \cite{Kin:02:JMP} and entanglement breaking channels \cite{Sho:02:JMP} are known to be additive and thus their classical capacity can be computed without the need of regularization. Similarly, for the task of classical communication over a quantum channel, one can prohibit the use of inputs states correlated over multiple uses of the channel---effectively allowing optimization on the individual channel use only---to obtain a lower bound on the classical capacity of a quantum channel. This notion of capacity is known as the Holevo capacity. Even with such a  simplification of the problem, the calculation remains considerably demanding. As a matter of fact, calculation of the Holevo capacity falls in the category of NP-complete problems \cite{BS:08:arXiv,ES:16:PRA}.

This multilayer difficulty has stimulated a good amount of research in the field of quantum information theory. Different researchers have taken different routes to accomplish this seemingly impossible task. 
 For example, different definitions of capacities have been proposed \cite{WY:16:IEIT}, analytical expressions for the special channels have been found \cite{Kin:03:IEEE_T_IT},  and some bounds that are additive and easier to calculate have been computed \cite{FG:17:IEEE_T_IT} to solve the problem of regularization. While for solving the difficulty of calculation, exploiting special properties of a given channel \cite{Cor:04:PRA}, and methods that can approximate the capacity upto a fixed a posteriori error have been proposed \cite{SSE:16:T_IT}.

In this work we give easy to compute lower and upper bounds on the Holevo capacity of discrete Weyl channels (DWCs). Our employed approach involves modeling the DWC as a classical symmetric channel and use the existing results from the classical information theory to lower bound the Holevo capacity of a DWC. The upper bound is based on the majorization relation of any possible output state of a DWC with the most ordered state based on the channel parameters. We give a necessary and sufficient condition for which the two bounds coincide. We find that this condition is met for the known special cases (Pauli qubit channel, and the qudit depolarizing channel) of DWC and hence we can recover the exact capacity expression for these cases. Through numerical examples we show that the coincidence of two bounds is sufficient but not necessary for the lower bound to give exact capacity.

\section*{Discrete Weyl Channel}

A quantum state $\pmb{\rho}$ on the Hilbert space is a positive operator with unit trace (i.e., a density operator). We consider the Hilbert space of finite dimension $d$. The state is said to be \emph{pure} if it has the form $\pmb{\rho} = \ketbra{\psi}{\psi}$. We usually denote a pure state simply by a ket e.g., $\ket{\psi}$, which is a column vector in the Hilbert space. A quantum channel
$
\mathcal{N}: \pmb{\rho} \rightarrow \qc{}{\pmb{\rho}}
$
is a completely positive trace preserving (CPTP) map transforming the input state $\pmb{\rho}$ to an output state $\qc{}{\pmb{\rho}}$. The map can be specified in terms of Kraus operators $\left\{\M{A}_i \right\}$ as
$
\qc{}{\pmb{\rho}} = \sum_{i} \M{A}_{i} \pmb{\rho} \M{A}_{i}^{\dag}
$
where $\sum_{i}\M{A}_{i}^{\dag}\M{A}_{i} = \M{I}_d$ and $\M{I}_d$ is the identity operator on the $d$-dimensional Hilbert space. For a random unitary channel, it is possible to represent Kraus operators as $\M{A}_{i} = \sqrt{p_i} \M{B}_{i}$, such that the channel applies an operator $\M{B}_{i}$ on the input state with the probability $p_i$ \cite{NC:11:CUP}.

\begin{figure}[t!]
\centering
\includegraphics[width = 0.49\textwidth]{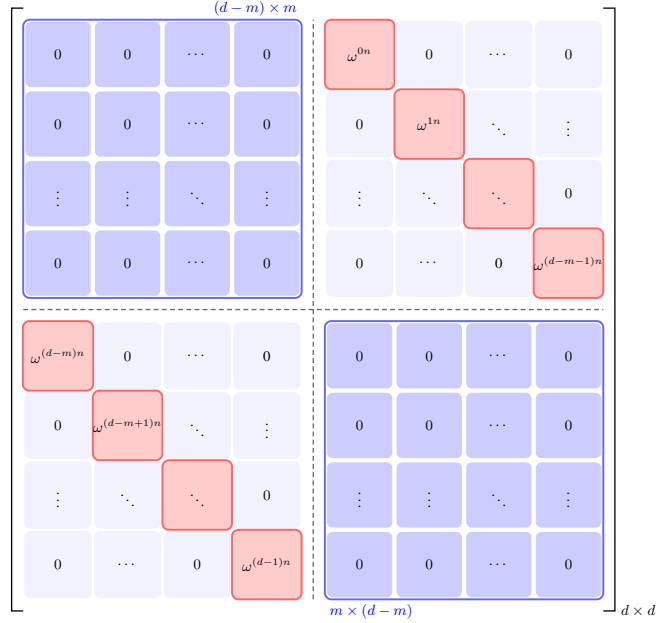}
\caption{The general structure of a Weyl operator $\M{W}_{nm}$ in an arbitrary dimension $d$.}
\label{fig:Weyl_str}
\end{figure}

Let $\pmb{\sigma}_0 = \M{I}_2$ be the $2\times 2$ identity matrix, and
\begin{align}
\pmb{\sigma} _{1} = \begin{bmatrix}
0&1\\1&0
\end{bmatrix}, \hspace{0.2cm}
\pmb{\sigma} _{2} = \begin{bmatrix}
0&-\imath\\\imath&0
\end{bmatrix}, \hspace{0.2cm}
\pmb{\sigma} _{3} = \begin{bmatrix}1&0\\0&-1\end{bmatrix}
\end{align}
be the Pauli matrices. The Pauli qubit channel, denoted by $\qc{\mathrm{p}}{\pmb{\rho}}$, is then defined as 
\begin{align}
\qc{\mathrm{p}}{\pmb{\rho}} = \sum_{i = 0}^3 p_{i}\pmb{\sigma}_{i} \pmb{\rho} \pmb{\sigma}_{i}^{\dag}
\end{align}
which is a random unitary channel.

Discrete Weyl operators are a non-Hermitian generalization of Pauli operators for dimension $d$ \cite{BK:08:JPAMT}. A Weyl operator $\M{W}_{nm}$ 
 on the $d$-dimensional Hilbert space is defined as \cite{Wey:27:QG}
\begin{align}
\M{W}_{nm} &= \sum_{k = 0}^{d - 1} \omega^{kn}\ketbra{k}{\modA{k}{m}{d}}
\label{def:WeylOper}
\end{align}
for $n,m = 0,1,\cdots , d-1$; $\omega = \exp \left(2\pi \imath/d \right)$; and $\ket{k}$ is the $k$th basis vector in the computational basis (for notational convenience, the indexing of entries of vectors and matrices start from 0). A general structure of a $d$-dimensional Weyl operator $\M{W}_{nm}$ is shown in Fig.~\ref{fig:Weyl_str}. 

\begin{property}
A Weyl operator $\M{W}_{nm}$, when applied on a $d$-dimensional vector $\ket{\alpha}$, up-shifts the entries of $\ket{\alpha}$ by $m$ locations and rotates $i$th entry (according to new indexing) by a phase of $\omega^{in}$. We refer to this property as shift and phase operation of Weyl operators.
\label{property:shift_phase}
\end{property}

Eigenvalues of a Weyl operator $\M{W}_{nm}$ are given by (see Methods)
\begin{align}
\lambda_s = \omega^{mn\frac{\left(d-1 \right)}{2} + s}
\label{eq:eig}
\end{align}
where $s \in  \left\{ \modS{mk}{nj}{d}\right\}$ for $j,k = 0,\cdots , d-1$. A schematic illustration for the Weyl operator $\M{W}_{31}$ on a 4-dimensional Hilbert space is given in Fig.~\ref{fig:eig_val}. Note that Weyl operators operating on a prime dimensional Hilbert space have $d$ distinct eigenvalues {(and we can simply state that $s=0,1,\cdots , d-1$)} except for $\M{W}_{00}$. On the other hand, some Weyl operators of a composite dimension may have repeated eigenvalues. This repetition of eigenvalues restrains us from deriving general forms of our results directly. We circumvent this problem by first presenting our results for the Hilbert space of a prime dimension, and then show that an alternate formulation of our results can be applied to the case of a composite dimensional Hilbert space as well.

\begin{figure*}[t!]
\centering
\includegraphics[width = 0.83\textwidth]{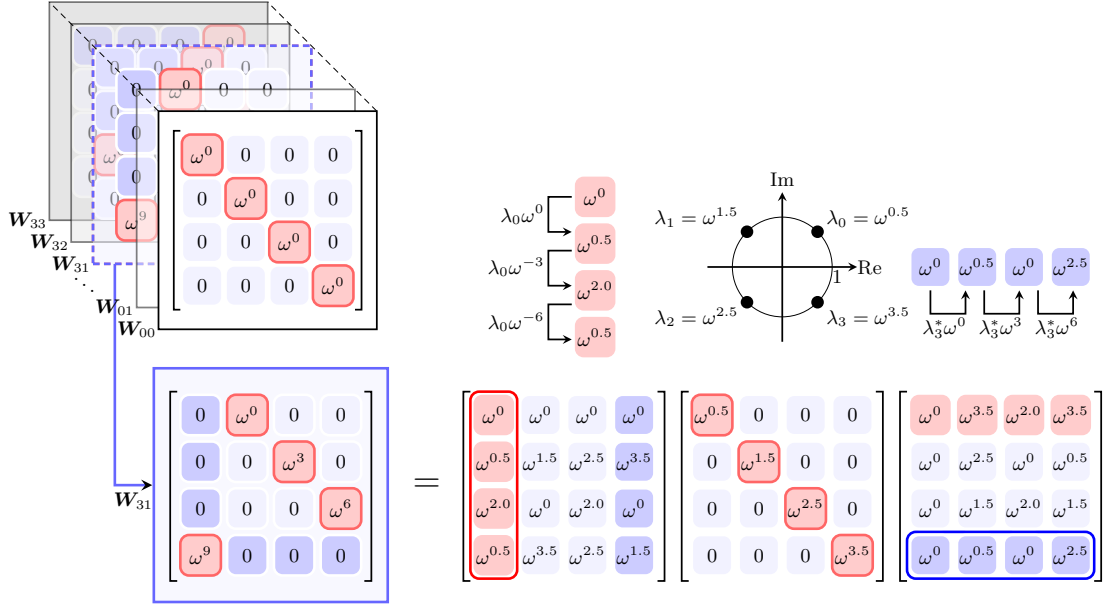}
\caption{A schematic illustration for the structure of discrete Weyl operator $\M{W}_{31}$ on a 4-dimensional Hilbert space. Each eigenvalue $\lambda_s$ and eigenvector $\ket{\lambda_s}$ can be found using \eqref{eq:eig} and \eqref{eq:AlphaRelation}, respectively. } 
\label{fig:eig_val}
\end{figure*}

A DWC, denoted by $\qc{\mathrm{dw}}{\pmb{\rho}}$, is a generalization of the Pauli qubit channel \cite{WIL:11:CUP}, defined in terms of discrete Weyl operators as
\begin{align}
\qc{\mathrm{dw}}{\pmb{\rho}} = \sum_{n = 0}^{d - 1}\sum_{m = 0}^{d - 1} p_{nm}\M{W}_{nm} \pmb{\rho}  \M{W}_{nm}^{\dag}
\label{eq:PauliQudit}
\end{align}
where $\M{W}_{nm}$ acts on the input state $\pmb{\rho}$ with probability $p_{nm}$.

The Holevo capacity of a quantum channel is defined as\cite{Hol:98:T_IT,SW:97:PRA}
\begin{align}
\chi\left(\mathcal{N}\right) = 
\sup_{\left\{ p_i, \pmb{\rho}_i\right\}} 
	\left[   
		S\left( \sum_{i} p_i \mathcal{N}\left(\pmb{\rho}_i\right) \right) - 
		\sum_i p_i S\left(\mathcal{N}\left(\pmb{\rho}_i \right) \right)
	\right]
	\label{eq:HolevoChi}
\end{align}
where $p_i$ is the \emph{a priori} probability of input state $\pmb{\rho}_i$; $S\left(\pmb{\rho}\right) = -\text{Tr} \left(\pmb{\rho} \log \pmb{\rho} \right)$ is the von Neumann entropy, and $\mathcal{N}\left( \pmb{\rho} \right)$ is the output state produced by the action of channel $\mathcal{N}$ on the input state $\pmb{\rho}$. The Holevo capacity corresponds to the maximum rate of classical information when  input states are restricted to be separable, i.e., the inputs of the channel are not entangled over multiple uses.

\begin{lemma}
If an input state of DWC operating on a $d$-dimensional Hilbert space is an eigenstate of a $d$-dimensional Weyl operator $\M{W}_{nm}$, then the output state is diagonal in the eigenbasis of $\M{W}_{nm}$.
\label{lemma:classical}
\end{lemma}
\begin{proof}
See Methods section.
\end{proof}

{As a consequence of the above Lemma, we can choose the set of input states to be $d$ orthogonal eigenvectors of some Weyl operator $\M{W}_{nm}$, and measure the output in the eigenbasis of $\M{W}_{nm}$. The uncertainty at the output of the channel in this case is purely classical in nature. In this sense, a DWC is behaving as a classical channel, transitioning a distinguishable state into an unknown but perfectly distinguishable state. We completely characterize the simulated classical channel in terms of channel transition matrix in the following Proposition.}

\begin{proposition}
{A DWC of a prime dimension $d$ with orthonormal eigenstates of $\M{W}_{nm}$ as the input states behaves as a classical symmetric channel with the following transition matrix}
\begin{align}
{
\M{T}_{nm} = \begin{bmatrix}
P_1 & P_2 &  \cdots & P_{d}\\
P_{d} & P_1 &  \cdots & P_{d-1}\\
\vdots &  \vdots & \ddots & \vdots\\
P_2 & P_3 &  \cdots & P_1
\end{bmatrix},
 \qquad \left(n,m\right) \neq \left(0,0\right)}
\label{eq:tr_matrix}
\end{align}
{where}
\begin{align}
{
P_k = \sum_{ij:\omega^{mi - nj} = \omega^{k-1}}p_{ij}.
}
\label{eq:pk_sum}
\end{align}

\label{prop:model}
\end{proposition}
\begin{proof}
See Methods section.
\end{proof}

\begin{figure*}[t!]
\centering
\includegraphics[width = 0.6\textwidth]{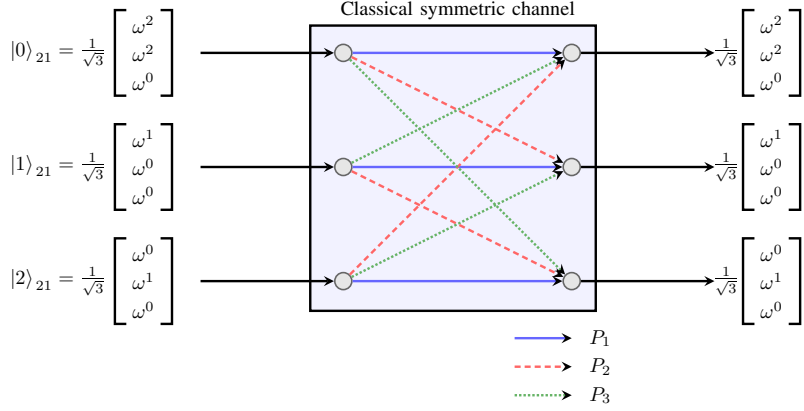}
\caption{An example DWC for $d=3$ driven by the eigenstates of $\M{W}_{21}$.  
}
\label{fig:dwc}
\end{figure*}

As an example, a DWC driven by the eigenstates of $\M{W}_{21}$ with $d = 3$ is shown in Fig.~\ref{fig:dwc}. In this example, we have 
$P_1= p_{00} + p_{21} + p_{12}$, 
$P_2= p_{20} + p_{11} + p_{02}$, and 
$P_3= p_{10} + p_{01} + p_{22}$.

\section*{Results}

Based on the proposition~\ref{th:pauliQudit}, we give the following simple and natural lower bound on the Holevo capacity of a DWC:
\begin{theorem}
The Holevo capacity $\chi\left(\qcS{dw}\right)$ of the channel in \eqref{eq:PauliQudit} with a prime $d$ is bounded as
\begin{align}
{
\chi\left(\qcS{dw}\right) \geq \log_2 \left(d\right) - \min_{n,m}H\left(\text{row of }\M{T}_{nm}\right), \qquad (n,m) \neq (0,0)
}
\label{eq:CapQudit}
\end{align}
where $\M{T}_{nm}$ is the channel transition matrix of the $(n,m)$th symmetric channel obtained by fixing the eigenstates of $\M{W}_{nm}$ as the signal states and $H\left( \cdot \right)$ is the Shannon entropy.
\label{th:pauliQudit}
\end{theorem}
\begin{proof}
See Methods section.
\end{proof}

{
The restriction on $d$ to be a prime number is primarily because the repetition of eigenvalues of $\M{W}_{nm}$ of a composite $d$ does not allow us to construct the channel transition matrix $\M{T}_{nm}$. The following remark provides us an alternative approach to lower bound the Holevo capacity of DWC of any $d$. 
}

\begin{remark}
{It is straightforward to show that $H\left(\text{row of } \M{T}_{nm}\right) = S\left( \qc{\mathrm{dw}}{\ket{\lambda}\bra{\lambda}_{nm}} \right)$ when $d$ is prime, where $\ket{\lambda}\bra{\lambda}_{nm}$ is the density matrix of any eigenstate of $\M{W}_{nm}$. Therefore, we can equivalently calculate}
\begin{align}
{
\chi\left(\qcS{\mathrm{dw}}\right) \geq \log_2 \left(d\right) - \min_{n,m }S\left( \qc{\mathrm{dw}}{\ket{\lambda}\bra{\lambda}_{nm}} \right)
}
\label{eq:l_bound_state}
\end{align}
{for prime $d$. Then, we can extend \eqref{eq:l_bound_state} to any $d$ by replacing the optimization on any $\pmb{\rho}$ with the optimization only on the eigenstates of $\M{W}_{nm}$ in \eqref{eq:Hard_Optimization}.}

\label{remark1}
\end{remark}

\begin{theorem}
Let us define a vector $\pmb{\zeta}\left(\V{p}\right) \in \mathbb{R}^{d}$ such that
\begin{align}
\pmb{\zeta}\left(\V{p}\right)=\M{S} \V{p}^{\downarrow} 
\end{align}
where the elements of $\V{p}^{\downarrow}$ are the elements of vector $\V{p} \in \mathbb{R}^{d^2}$ in descending order; the matrix $\M{S} \in \mathbb{R}^{{d} \times d^2}$ is given by
\begin{align}
\M{S} = 
\begin{bmatrix}
\V{1}_d^T & \V{0}_d^T & \V{0}_d^T & \cdots & \V{0}_d^T\\
\V{0}_d^T & \V{1}_d^T & \V{0}_d^T & \cdots & \V{0}_d^T\\
\V{0}_d^T & \V{0}_d^T & \V{1}_d^T & \cdots & \V{0}_d^T\\
\vdots & \vdots & \vdots & \ddots & \vdots \\
\V{0}_d^T & \V{0}_d^T & \V{0}_d^T & \cdots & \V{1}_d^T\\
\end{bmatrix}
\end{align}
where $\left( \cdot \right)^T$ denotes the transpose operation, and $\V{1}_d$ and $\V{0}_d$ are all-one and all-zero vectors of $d$ elements, respectively.
Then, the Holevo capacity of a DWC is
\begin{align}
\chi\left(\qcS{dw}\right) \leq \log_2 \left(d\right) - H\left( \pmb{\zeta} \left( \V{p}\right)\right),
\label{eq:u_bound}
\end{align}
where $\V{p}=\left[p_{00}~~p_{01}~~ \cdots ~~p_{nm}\right]^T$, whose elements are probabilities associated with respective Weyl operators $\M{W}_{nm}$.
\label{th:u_bound}
\end{theorem}
\begin{proof}
See Methods section.
\end{proof}

\

In a $d$-dimensional Hilbert space, $d^2$ Weyl operators are defined whose indices are given in the form of 2-tuples, e.g., $(i,j)$. We define a set $\mathcal{W}$ that contains all the $d^2$ indices of defined Weyl operators. We call a set $\mathcal{D}$ a $d$-set if all its elements $\mathcal{D}_i $ for $i = 0,\cdots , d-1$ are non-overlapping $d$ element subsets of $\mathcal{W}$
\begin{align}
\mathcal{D} &= \left\{\mathcal{D}_i | \mathcal{D}_i \subset_d \mathcal{W}, \mathcal{D}_i \cap \mathcal{D}_j = \varnothing \text{ for } i \neq j, i,j = 0,\cdots , d-1 \right\}
\end{align}
where $\mathcal{A}\subset_d \mathcal{B}$ means that $\mathcal{A}$ is a $d$-element subset of $\mathcal{B}$, $\varnothing$ is the empty set, $\mathcal{A} \cap\mathcal{B}$ gives a set whose elements are the common elements of $\mathcal{A}$ and $\mathcal{B}$. In the $d$ dimensional Hilbert space, there are 
\begin{align}
\frac{1}{d!}\prod_{i = 0}^{d-1} \begin{pmatrix}
d^2 - id \\ d
\end{pmatrix}
\label{eq:n_d-sets}
\end{align}
different possible $d$-sets, where 
\begin{align*}
\begin{pmatrix}
n\\k
\end{pmatrix} = 
\frac{n!}{k! \left(n - k\right)!}
\end{align*}
are the binomial coefficients.

A $d$-set $\mathcal{D}$ whose all elements $\mathcal{D}_i$ satisfy the property
\begin{align}
mi - nj \text{ mod } d = k_i, \qquad \forall \left( i,j\right) \in \mathcal{D}_i
\label{eq:Achiev_property}
\end{align}
for some $n,m$, and some constants $k_i$ is called an achievable $d$-set. For example 
\begin{align}
\mathcal{D} &= \left\{
				\left\{
					(0,0),(2,1),(1,2) 
				\right\},
				\left\{
					(2,0),(1,1),(0,2) 
				\right\}, 
				\left\{
					(1,0),(0,1),(2,2) 
				\right\} 
				\right\}
\end{align}
is an achievable $d$-set for $\left(n,m\right) = \left(2,1 \right)$ but 
\begin{align}
\mathcal{D} &= \left\{
				\left\{
					(0,0),(0,1),(1,2) 
				\right\},
				\left\{
					(2,0),(1,1),(2,2) 
				\right\},  
				\left\{
					(1,0),(2,1),(0,2) 
				\right\} 
				\right\}
\end{align}
is a $d$-set which is not achievable.
\begin{theorem}
We arrange the elements of $\V{p}$ in nonincreasing order and collect the indices of $p_{nm}$ while preserving the order to form a $d$-set. The bounds of Theorem~\ref{th:pauliQudit}, and Theorem~\ref{th:u_bound} coincide if and only if (resp. only if) the obtained $d$-set is achievable and $d$ is a prime number (resp. a composite number). 
\label{th:Suff}
\end{theorem}
\begin{proof}
See Methods section.
\end{proof}

\begin{remark}
 If the two bounds coincide, we have
\begin{align}
{
\chi\left(\qcS{\mathrm{dw}}\right) = \log_2 \left(d\right) - \min_{n,m}H\left( \text{row of }\M{T}_{nm}\right), \qquad \left( n,m\right)\neq \left( 0,0\right).
}
\label{eq:CapQuditExact}
\end{align}
 However, the converse is not true as will be shown by the numerical examples in the next section.
\end{remark}

\section*{Discussion}
An efficient approximation for the capacity of classical-quantum channels has been discussed without exploiting any special properties of a given channel. For example, it takes 40,154 seconds in order to approximate the Holevo capacity of a Pauli qubit channel with a posteriori error of $\num{1.940e-3}$ \cite{SSE:16:T_IT}. In contrast to existing methods, the average time to calculate the (lower) bound  in this paper is of the order $10^{-4}$ seconds even for large $d$ by virtue of the use of special properties of DWCs.

We have strong numerical evidence that the lower bound is tighter and is saturated more often even when the two bounds do not coincide, as shown in the Figs.~\ref{fig:rand_chann_main}(a), \ref{fig:rand_chann_main}(b), and \ref{fig:rand_chann_main}(c) where the upper ($\chi_{\text{UB}}$) and the lower ($\chi_{\text{LB}}$) bounds (normalized by $\log_2 \left( d \right)$) are plotted for 1200 random channel realizations for $d=3, 4$, and $5$, respectively. In these figures, Holevo capacity by using \cite{Cor:04:PRA}
\begin{align}
\chi \left( \mathcal{N}_{\mathrm{dw}}\right)= \log_2 \left(d\right) - \min_{\pmb{\rho}}S\left( \mathcal{N}_{\mathrm{dw}}\left(\pmb{\rho}\right)\right)
\label{eq:Hard_Optimization}
\end{align} 
with the optimization performed via genetic algorithm $(\chi_{\text{GA}})$ is also presented. Comparison of $\chi_{\text{LB}}$, $\chi_{\text{UB}}$, and $\chi_{\text{GA}}$ shows that the frequency of coincidence of two bounds as well as the frequency of the saturation of the lower bound is higher for the case of $d=3$.

\begin{figure}[t!]
\centering
\includegraphics[width = 1\textwidth]{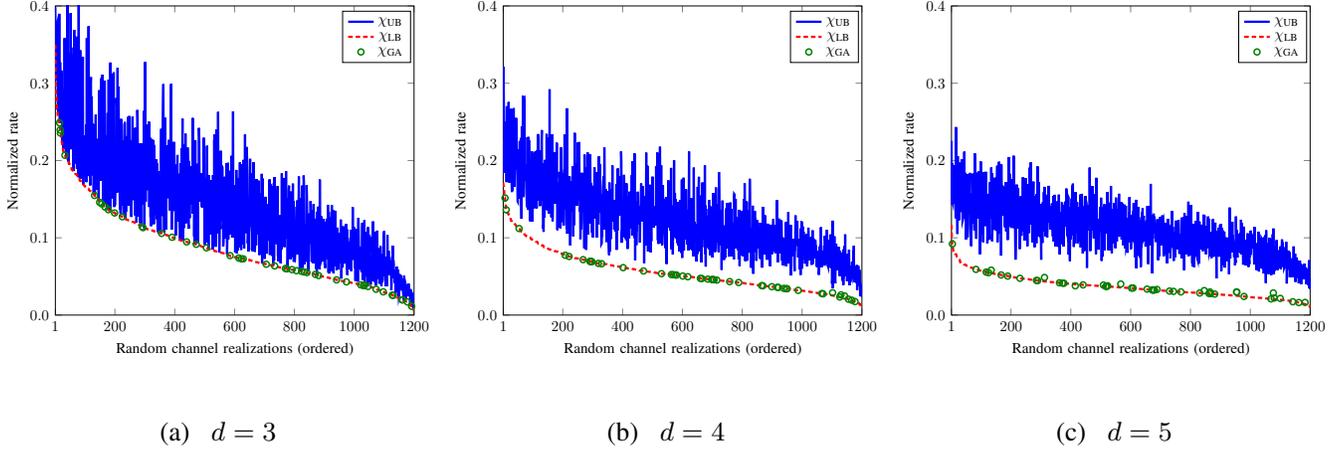}

\caption{$\chi_{\text{UB}}$, $\chi_{\text{LB}}$, and $\chi_{\text{GA}}$ of random channel realizations (in decreasing order of $\chi_{\text{UB}}$) when $d=3,4,5$. }
\label{fig:rand_chann_main}
\end{figure}
Our bounds not only ease the requirement of optimization for the calculation of tight bounds for a general DWC, but also allows to recover the analytic expressions for the special cases of DWC. 
For example, here we recover the analytic expression for the classical capacity of a qudit depolarizing channel using the approach developed above. A quantum depolarizing channel transforms an input state to the output state according to the following map
\begin{align}
\qc{\mathrm{d}}{\pmb{\rho}} = \left(1 - \mu\right)\pmb{\rho} + \mu \pmb{\pi}
\label{def:Depolarizing}
\end{align}
where $\pmb{\pi} = \M{I}_d/d$ is the maximally mixed state on the output Hilbert space. 

 In terms of Weyl operators,
\begin{align}
\pmb{\pi} = \frac{1}{d^2}\sum_{n,m = 0}^{d - 1} \M{W}_{nm}\pmb{\rho} \M{W}_{nm}^{\dag}.
\end{align}
Thus, we can rewrite  equation~\eqref{def:Depolarizing} as
\begin{align}
\qc{\mathrm{d}}{\pmb{\rho}} = \left(1 - \mu + \frac{\mu}{d^2}\right)\pmb{\rho} + \frac{\mu}{d^2}\sum_{\substack{ n,m = 0 \\ \left( n,m\right) \neq \left( 0,0\right) }}^{d - 1} \M{W}_{nm}\pmb{\rho} \M{W}_{nm}^{\dag}.
\end{align}
Therefore 
\begin{align}
p_{00} = 1 - \mu + \frac{\mu}{d^2}, \quad p_{nm} = \frac{\mu}{d^2} \quad \forall\, \left(n,m\right) \neq \left(0,0\right) 
\end{align} 
which shows that all $d$-sets (whether achievable or not) are equivalent in terms of summation of $p_{nm}$ over the elements $\mathcal{D}_i$. Therefore, we can choose an ordering of $p_{nm}$ such that the condition of Theorem~\ref{th:Suff} is satisfied and we can use  equation~\eqref{eq:u_bound} to calculate the Holevo capacity. From  equation~\eqref{def:Depolarizing} and the output vector of $\pmb{\zeta}\left(\V{p}\right) = { \left(r_0, r_1, \cdots, r_d \right)}$, we see that
\begin{align}
r_0 = 1 - \mu + \frac{\mu}{d}, \quad r_i = \frac{\mu}{d} \quad \text{for} \quad i = 1,\cdots , d-1.
\end{align}
Thus, the Holevo capacity $\chi \left( \mathcal{N}_\mathrm{d}\right)$ of this channel is
\begin{align}
\chi \left( \mathcal{N}_{\mathrm{d}}\right) &= \log_2\left(d\right) + \left(1 - \mu + \frac{\mu}{d} \right)\log_2\left(1 - \mu + \frac{\mu}{d} \right) + \left(d - 1\right)\frac{\mu}{d}\log_2\left(\frac{\mu}{d}\right)
\label{eq:HC_Dep}
\end{align}
which is equal to the classical capacity of the quantum depolarizing channel \cite{Kin:03:IEEE_T_IT}.

Additionally, it is easy to see that for a Pauli qubit channel ($d=2$), there are 3 possible $d$-sets which are all achievable. Therefore, both bounds are exact for the Pauli qubit (and \emph{all its special cases}) channel. With simple algebraic manipulations one can obtain the analytic expressions for the capacities of any of the special cases of the Pauli qubit channel \cite{SSE:16:T_IT}.

From Theorem~\ref{th:Suff}, we can also define special channels for which the two bounds always coincide. This approach gives us a class of quantum channels whose exact Holevo capacity can readily be calculated. We define two such channels here and call them one-parameter depolarizing like, and two-parameter depolarizing like channels, respectively.

The one-parameter depolarizing-like channel is defined as
\begin{align}
\qc{\mathrm{d}1}{\pmb{\rho}} = \left(1 - \xi \right) \M{W}_{ij}\pmb{\rho} \M{W}_{ij}^{\dagger} + \xi\pmb{\pi},
\end{align}
whose exact Holevo capacity is same as \eqref{eq:HC_Dep} with the depolarizing parameter $\xi$.

The two-parameter depolarizing-like channel is 
\begin{align}
\qc{\mathrm{d}2}{\pmb{\rho}} &= \left(1 - \eta \right)\M{W}_{ij}\pmb{\rho} \M{W}_{ij}^{\dagger} + \left(1 - \kappa \right)\M{W}_{nm}\pmb{\rho}\M{W}_{nm}^{\dagger} + \left(\eta + \kappa - 1 \right) \pmb{\pi}
\label{def:two_pauli}
\end{align}
 where $0 \leq \eta, \kappa \leq 1 , \text{and } 1 \leq \eta + \kappa \leq 2$.
 This channel is a further generalization of the one-parameter depolarizing like channel. The exact Holevo capacity of this channel can readily be calculated by Theorem~\ref{th:Suff}.

In this work we modeled a DWC as a classical symmetric channel for the task of classical communication. Through this modeling, we presented a simple to compute lower bound on the Holevo capacity of a given DWC of an arbitrary dimension. We also gave an intuitive upper bound which coincides with the lower bound under a certain condition. This (sufficient and necessary for a prime $d$, and necessary for a composite $d$) condition, however, is not frequently met despite the frequent convergence of the lower bound to the actual Holevo capacity as shown by the numerical examples.
 The lower bound was derived by noting the similarity of a quantum channel with a classical channel. An interesting future direction is to find similar cases where the results of classical information theory (which is more mature despite being a special case of quantum information theory) can be applied on the problems of quantum information theory with a little or no modification. Similarly, based on the equality of upper and lower bounds, one can define special channels for which these bounds always coincide. Such characterization of quantum channels can give us a class of channels whose exact Holevo capacity can readily be calculated.

\section*{Methods}
\subsection*{Proof of Lemma~\ref{lemma:classical}}
Since the DWC is a random unitary channel, the output of the channel is merely the state obtained by randomly applying one of the $d^2$ Weyl operators on the input. Thus, we need to show that operation of $\M{W}_{ij}$ on any eigenstate of $\M{W}_{nm}$ results into an eigenstate of $\M{W}_{nm}$.

Let
\begin{align}
\ket{\lambda} = \begin{bmatrix}
\alpha_0, \alpha_1,  \cdots , \alpha_{d - 1}
\end{bmatrix}^T
\label{eq:WeylEigenVector}
\end{align}
be a normalized eigenvector of $\M{W}_{nm}$ with the corresponding eigenvalue $\lambda$. From the eigenvalue relation $\M{W}_{nm}\ket{\lambda} = \lambda \ket{\lambda}$, and due to the property~\ref{property:shift_phase}, we get the following relation among the entries of vector of \eqref{eq:WeylEigenVector}
\begin{align}
\alpha_{\modA{m}{k}{d}} &= \lambda\omega^{-nk}\alpha_k,
\label{eq:AlphaRelation}
\end{align}
where the eigenvalues $\lambda$ are equidistant points on the unit circle (see Fig.~\ref{fig:eig_val}). Since we have obtained this relation from the condition of eigenvector, any vector satisfying above relation will be an eigenvector of $\M{W}_{nm}$.

Now let us consider the effect of any $\M{W}_{ij}$ on the vector of \eqref{eq:WeylEigenVector}. To this end, we let $\M{W}_{ij}\ket{\lambda} = \ket{\beta}$, and recall property~\ref{property:shift_phase} again to write
\begin{align}
\M{W}_{ij}\ket{\lambda}  
= \begin{bmatrix}
\alpha_{j} \\ \omega^{i}\alpha_{\modA{j}{1}{d}} \\ \vdots \\ \omega^{ki}\alpha_{\modA{j}{k}{d}} \\ \vdots
\end{bmatrix}
= \begin{bmatrix}
\beta_0 \\ \beta_1 \\ \vdots \\ \beta_{k} \\ \vdots
\end{bmatrix} 
= \ket{\beta}.
\end{align}
i.e., the $k$th entry of $\ket{\beta}$ is $\omega^{ki}\alpha_{\modA{j}{k}{d}}$.

If the elements of $\ket{\beta}$ exhibit a similar relation as \eqref{eq:AlphaRelation}, $\ket{\beta}$ is also an eigenvector of $\M{W}_{nm}$.
Repeated use of \eqref{eq:AlphaRelation} gives the following relation between the entries of $\ket{\beta}$
\begin{align}
\beta_{\modA{m}{k}{d}} = \lambda\omega^{mi - nj}\omega^{-nk}\beta_{k}
\end{align}
which essentially bears the same form as  \eqref{eq:AlphaRelation}; because $\lambda\omega^{mi - nj}$ is another eigenvalue of $\M{W}_{nm}$. Hence the vector $\ket{\beta} = \M{W}_{ij} \ket{\lambda}$ is an eigenvector of $\M{W}_{nm}$. Since the output state is a statistical mixture of orthonormal eigenstates of $\M{W}_{nm}$, it is diagonal in the same basis, i.e., in the eigenbasis of $\M{W}_{nm}$.


\subsection*{Proof of Proposition~\ref{prop:model}}
Let the input state be an eigestate $\ket{\lambda}$ of $\M{W}_{nm}$ corresponding to the eigenvalue $\lambda$. From the proof of Lemma~\ref{lemma:classical}, the application of $\M{W}_{ij}$ transforms the input state to the eigenstate of $\M{W}_{nm}$ corresponding to the eigenvalue $\lambda\omega^{mi - nj}$. Since $\omega = \exp \left(2\pi \iota / d \right)$, $\omega^{mi - nj}$ is always from the set $ \left\{\omega^0, \omega^1, \cdots , \omega^{d-1} \right\}$. Therefore, we can define,
\begin{align}
P_k = \sum_{ij:\omega^{mi - nj} = \omega^{k-1}}p_{ij}
\end{align}
as the transition probability of $\ket{\lambda}$ to the orthogonal state $\ket{\lambda\omega^{k-1}}$. {We can define the complete set of transition probabilities $P_k$, for  $k = 1,2, \cdots , d$ only if $\M{W}_{nm}$ does not have any repeated eigenvalues which is guaranteed only if $d$ is prime and $\left(n,m \right) \neq \left( 0,0 \right)$ (note the similarity between $\omega^{mi - nj}$ and the expression for $s$ in the definition of eigenvalues).}

Furthermore, we notice that the rows of $\M{T}_{nm}$ are permutations of each other and its columns are permutation of each other. Therefore, $\M{T}_{nm}$ in \eqref{eq:tr_matrix} defines a classical symmetric channel. 

\subsection*{Proof of Theorem~\ref{th:pauliQudit}}
From proposition~\ref{prop:model} we know that in this setting DWC acts as a classical symmetric channel. Since the capacity of a symmetric channel with $d$ inputs and outputs is given by \cite{CJ:12:JWS}
\begin{align}
{
C_{\text{Symmetric}} = \log_2 \left(d\right) - H\left( \text{row of transition matrix}\right),
}
\label{eq:CapSymmetric}
\end{align}
and we have restricted our input states to be from the eigenstates of Weyl operators, thus 
\begin{align*}
{
\chi\left(\qcS{dw}\right) \geq \log_2 \left(d\right) - \min_{n,m}H\left(\text{row of }\M{T}_{nm}\right), \qquad (n,m) \neq (0,0)
}
\end{align*}
{
where the condition $(n,m) \neq (0,0)$ along with the condition on $d$ to be prime ensures that we can model the given DWC as a classical symmetric channel with the channel transition matrix $\M{T}_{nm}$ by virtue of Proposition~\ref{prop:model}.
}

\subsection*{Proof of Theorem~\ref{th:u_bound}}
We can write \eqref{eq:PauliQudit} as
\begin{align}
\dw \left( \rho\right) &= \sum_{j = 1}^{d^2} q_j V_j \pmb{\rho} V_j^{\dagger} 
\label{eq:DWC_def_q}\\
&= \sum_{i = 1}^{d} \lambda_i \ket{\lambda_i}\bra{\lambda_i},
\label{eq:DWC_eigen}
\end{align}
where the vector $\V{q} = \left[ q_1, q_2, \cdots, q_{d^2}\right]$ is the vector of elements of $\V{p} = \left[p_{0,0}, p_{0,1}, \cdots, p_{d-1, d - 1} \right]$ arranged in descending order, denoted by $\V{q} = \V{p}^{\downarrow}$, and $V_j$ is the Weyl operator corresponding to $q_j$, i.e., $q_j = p_{n,m}\implies V_j = W_{n,m}$. The last equality is the eigendecomposition of $\dw \left( \rho\right)$, where $\lambda_j$ is the $j$th largest eigenvalue with $\ket{\lambda_i}$ being the corresponding normalized (unit norm) eigenvector. 

Let us denote by $D_{\pmb{\rho}} \left( \cdot \right)$ the mapping 
$$
D_{\pmb{\rho}}: \V{p} \rightarrow \V{\lambda}\left( \dw \left( \pmb{\rho}\right) \right),
$$
where $\pmb{\rho}$ is a normalized (unit trace) pure state input to the DWC, and $\V{\lambda}\left( \pmb{\sigma} \right)$ is the vector of eigenvalues of $\pmb{\sigma}$ in descending order. Then we can claim:
\begin{itemize}
	\item[1.] $D_{\pmb{\rho}}\left( \V{p}\right) = \M{T}_{\pmb{\rho}} \V{q}$, where $\V{q} = \V{p}^{\downarrow}$, and the $(i,j)$th element $T_{i,j}\in \left[ 0,1\right]$ of the matrix $\M{T}_{\pmb{\rho}}\in \mathbb{R}^{d\times d^2}$ is defined as
	\begin{align}
	T_{i,j} = \braket{\lambda_i|\pmb{\rho}_j|\lambda_i},
	\label{eq:T_def}
	\end{align}
	where $\pmb{\rho}_j = \M{V}_j \pmb{\rho} \M{V}_j$.\\
	\emph{Proof.} Since,
	\begin{align}
	\lambda_i &= \braket{\lambda_i|	\dw \left( \pmb{\rho}\right)| \lambda_i}\\
	&= \sum_{j = 1}^{d^2} q_j \braket{\lambda_i|\pmb{\rho}_j|\lambda_i}.
	\end{align}
	Therefore,
	\begin{align} 
	\begin{bmatrix}
	\lambda_1\\ \lambda_2 \\ \vdots \\ \lambda_d
	\end{bmatrix}
	= \begin{bmatrix}
	T_{1,1} & T_{1,2} & \cdots & T_{1,d^2}\\
	T_{2,1} & T_{2,2} & \cdots & T_{2,d^2}\\
	\vdots & \vdots & \ddots & \vdots\\
	T_{d,1} & T_{d,2} & \cdots & T_{d,d^2}
	\end{bmatrix}
	\begin{bmatrix}
	q_1 \\ q_2 \\ \vdots \\ q_d^2
	\end{bmatrix}.
	\end{align}
	\item[2.] Each column of $\M{T}$ sums to one, each row of $\M{T}$ sums to $d$.\\
	\emph{Proof.} Summation on the $j$th column is
	\begin{align}
	\sum_{i = 1}^{d} T_{i,j} &= \sum_{i = 1}^{d} \braket{\lambda_i|\pmb{\rho}_j|\lambda_i}\\
	& = \mathrm{tr}\left( \pmb{\rho}_j\right)\\
	&= 1.
	\end{align} 
	The first equality is from the definition of $T_{i,j}$, second equality follows from the fact that the set of  eigenvectors form a complete orthonormal basis, third from the fact that $\M{V}_i$ are unitary operators and the input $\pmb{\rho}$ has the unit trace. \\
	Summation on the $i$th row is
	\begin{align}
	\sum_{j = 1}^{d^2} T_{i,j} &= \sum_{j = 1}^{d^2} \braket{\lambda_i|\pmb{\rho}_j|\lambda_i}\\
	&= \braket{\lambda_i|\sum_{j = 1}^{d^2}\pmb{\rho}_j| \lambda_i}\\
	&= d \braket{\lambda_i|\M{I}| \lambda_i}\\
	&= d,
	\end{align}
	where $\M{I}$ is the identity matrix. The first equality follows again from the definition of $T_{i,j}$, the second equality from left and right distributive property of matrix products, third from the fact that
	\begin{align}
	\sum_{j = 1}^{d^2}\rho_j &= \sum_{j = 1}^{d^2}V_j\rho V_j\\
	&= d\M{I}.
	\end{align}
	Last equality is the consequence of $\ket{\lambda_i}$ having the unit norm. 
	\item[3.] Finally, the majorization relation $\V{s} = \M{S}\V{p}^{\downarrow} \succ \M{T}_{\pmb{\rho}} \V{p}^{\downarrow} = \V{t}$ holds for any pure state $\pmb{\rho}$.\\
	\emph{Proof.} First note that $\M{S}$ satisfies the conditions of being a valid $\M{T}_{\pmb{\rho}}$, i.e., $S_{i,j}\in [0,1]$, columns sum to 1, and rows sum to $d$. Also, by definitions of $\M{S}$ and $\M{T}_{\pmb{\rho}}$ the elements of $\V{s}$ and $\V{t}$ are already in the descending order, so no ordering is required on these vectors. The majorization relation is true if
	\begin{align}
	\sum_{\ell = 1}^{k} s_{\ell} \geq \sum_{\ell = 1}^{k} t_{\ell}, \quad \text{for } k = 1, \cdots, d,
	\label{eq:Maj_cond_1}
	\end{align}
	with equality for $k = d$.
	We can write
	\begin{align}
	\sum_{\ell = 1}^{k} s_{\ell} 
	&= \sum_{\ell = 1}^{k} \sum_{j = 1}^{d^2} S_{\ell,j}q_j \label{eq:S_sum}\\
	&= q_1 + q_2 + \cdots + q_{kd},
	\end{align}
	where $S_{\ell,j}$ is the $(\ell,j)$th element of $S$. Similarly,
	\begin{align}
	\sum_{\ell = 1}^{k} t_{\ell} 
	&= \sum_{\ell = 1}^{k} \sum_{j = 1}^{d^2} T_{\ell,j}q_j.\label{eq:T_sum}
	\end{align}	
	Both \eqref{eq:S_sum} and \eqref{eq:T_sum} can be seen as the weighted sum of $q_j$, where the maximum weight of each element ($T_{\ell,j}\leq 1$) as well as the sum of the weights ($\sum_{j = 1}^{d^2}T_{\ell,j} = d$) is fixed.\\
	We can see that the left hand side of \eqref{eq:Maj_cond_1} is the sum of $kd$ largest elements of $\V{q}$, i.e., maximum possible weights have been assigned to the largest elements of $\V{q}$. On the other hand, for any valid $\M{T}_{\pmb{\rho}}\neq \M{S}$, on the right hand side some weightage has been taken away from the larger elements and distributed among the smaller elements.	Since for any nonnegative real numbers $a_1 \geq a_2$, we trivially have
	$$
	a_1 \geq w a_1 + \left( 1 - w\right) a_2.
	$$ 
	Therefore, inequality in \eqref{eq:Maj_cond_1} holds for all $k = 1, \cdots, d$, for any valid $\M{T}_{\pmb{\rho}}$. The equality for $k = d$ holds because the left hand side is the sum of all $q_j$, and the right hand side
	\begin{align}
	\sum_{j = 1}^{d^2} q_j\sum_{\ell = 1}^{d}T_{\ell,j} = \sum_{j = 1}^{d^2} q_j
	\end{align}
	is also the sum of all $q_j$. Hence the majorization relation 
	\begin{align}
	\V{s}\succ \V{t}
	\label{eq:Main_major}
	\end{align} 
	holds for any pure input state $\pmb{\rho}$.
\end{itemize}
From \eqref{eq:Main_major}, it follows that
\begin{align}
\V{\zeta} \left( \V{p} \right) \succ \V{\lambda}\left( \dw \left( \pmb{\rho}\right) \right),
\label{eq:Main_major_2}
\end{align}
for any \emph{pure} input state $\pmb{\rho}$. Since the pure states are optimal for achieving the capacity \cite[Theorem~13.3.2]{WIL:11:ARX}, we can state more broadly that \eqref{eq:Main_major_2} holds for any input state $\rho$.\\
Finally, the main claim of the Theorem~\ref{th:u_bound}, inequality \eqref{eq:u_bound}, follows from the Schur concavity of Shannon entropy.


\subsection*{Proof of Theorem~\ref{th:Suff}}

{
We first observe that the condition on the summation in \eqref{eq:pk_sum} for the lower bound, and the condition on a $d$-set to be achievable \eqref{eq:Achiev_property} are essentially the same and result in the same $d$-element partitioning and ordering of $p_{nm}$. Thus, in a prime dimension $d$, every achievable $d$-set corresponds to a classical symmetric channel that can be simulated by DWC for some $n,m$.
}

{
On the other hand, the upper bound is obtained by ordering the elements of $p_{nm}$ in a nonincreasing order. Therefore, the achievability of the $d$-set formed by the indices of $p_{nm}$ when the $p_{nm}$ are arranged in a nonincreasing order is sufficient for the existence of a simulated classical symmetric channel of prime dimension that achieves the upper bound. Similarly, since the correspondence of achievable $d$-sets to a simulated classical symmetric channel is bijective, therefore the conincidence of two bounds necessarily implies the achievability of the $d$-set formed above.
}

{
For a composite $d$, the correspondence between the simulated classical symmetric channel to the achievable $d$-sets is injective-only. Therefore the above condition is necessary but no longer sufficient for the coincidence of two bounds.
}

\subsection*{Derivation of Eigenvalues of Discrete Weyl Operators}

The eigenvalues of a Weyl operator $\M{W}_{nm}$ are given by 
\begin{align}	\label{eq:eig_1}
\lambda_s = \omega^{mn\frac{\left(d-1 \right)}{2} + s}
\end{align}
where $s \in  \left\{ \modS{mk}{nj}{d}\right\}$ for $j,k = 0,\cdots , d-1$. Note that Weyl operators operating on a prime dimensional Hilbert space have $d$ distinct eigenvalues {(and we can simply state that $s=0,1,\cdots , d-1$)} except for $\M{W}_{00}$. On the other hand, some Weyl operators of a composite dimension may have repeated eigenvalues.

\subsubsection*{Sketch of the Proof}

The sketch of the proof is as follows. We use a previously known result \cite{{BHN:07:JPA}} to obtain an equation that has exactly $\ell$ distinct solutions, and all these $\ell$ solutions are the eigenvalues of a Weyl operator $\M{W}_{nm}$. Then, we show that \eqref{eq:eig_1} generates all $\ell$ solutions of the said equation.

It is shown in \cite[Theorem~4]{BHN:07:JPA} that the distinct eigenvalues $\tilde{\nu}_k$ (upto an appropriate phase factor) of Weyl operators operating on a $d$-dimensional Hilbert space are given by  
\begin{align}
\tilde{\nu}_k = \exp\left(2\pi \iota k/\ell\right),
\label{eq:savior}
\end{align}
for $0\leq k \leq \ell - 1$, where $\ell$ is either equal to $d$ or is some divisor of $d$. Let $p'$ be the said phase factor, then the distinct eigenvalues $\nu_k$ (with the exact phase) of a Weyl operator $\M{W}_{nm}$ are given by 
\begin{align}
\nu_k = p'\exp\left(2\pi \iota k/\ell\right).
\label{eq:eig_ref}
\end{align}
Note that \eqref{eq:eig_ref} are $\ell$th roots of some $p = \left( p'\right)^\ell$. In the following we derive some properties of $\ell$, and obtain an explicit expression for $p$. Then, we show that \eqref{eq:eig_1} generates the same eigenvalues as \eqref{eq:eig_ref} with the correct phase $p'$.

\subsubsection*{Determining $\ell$ and $p$} 

If an operator $\M{A}$ has eigenvalues $\left\{\mu_1, \mu_2, \cdots , \mu_n\right\}$, then the eigenvalues of $\M{A}^x$ are $\left\{\mu_1^x, \mu_2^x, \cdots , \mu_n^x\right\}$ \cite{Str:16:WCP}. Since the eigenvalues of $\M{W}_{nm}$ are all $\ell$th roots of $p$, $\left(\M{W}_{nm}\right)^{\ell}$ has only one eigenvalue i.e., $p$. Combining this fact with the fact that $\M{W}_{nm}$ are full rank matrices, and any similarity transform of identity results into identity, we deduce that 
\begin{align}
\left(\M{W}_{nm}\right)^{\ell} = pI.
\label{eq:ell_power}
\end{align}
Furthermore, since there does not exist any $\ell' < \ell$, such that $\left(\nu_k\right)^{\ell'} = p''$, for some $p''$ and for all $0\leq k \leq \ell - 1$, therefore $\ell$ is the smallest number such that $\left(\M{W}_{nm}\right)^{\ell}$ is proportional to $I$. The explicit expression for $\left(\M{W}_{nm}\right)^{q}$ for any integer $q$ is obtained in the following Lemma.

\begin{lemma}
	For any integer $q$ and a Weyl operator $\M{W}_{nm}$
	\begin{align}
	\left(\M{W}_{nm}\right)^{q} = \sum_{k = 0}^{d-1} \omega^{\left(q k + \frac{q\left(q - 1\right)}{2} m \right)n}\ket{k}\bra{\modA{k}{q m}{d}}.
	\label{eq:lemma_power}
	\end{align}
	\label{lemma:l_power}
\end{lemma}
\begin{proof}
	We prove this result by induction. Since Weyl operator $\M{W}_{nm}$ operating on a $d$-dimensional Hilbert space is defined as 
	\begin{align}
	\M{W}_{nm} = \sum_{k = 0}^{d - 1} \omega^{kn}\ket{k}\bra{\modA{k}{m}{d} },
	\label{eq:def:Weyl}
	\end{align}
	we have
	\begin{align}
	\left(\M{W}_{nm}\right)^2 &= \sum_{k = 0}^{d - 1} \omega^{kn}\ket{k}\bra{\modA{k}{m}{d} } \sum_{j = 0}^{d - 1} \omega^{jn}\ket{j}\bra{\modA{j}{m}{d} },\\
	&= \sum_{k = 0}^{d - 1} \omega^{\left(2k + m\right)n}\ket{k}\bra{\modA{k}{2m}{d} },
	\label{eq:def:Weyl_2}
	\end{align}
	where we have used the orthonormality of basis vectors.
	
	Now we assume that \eqref{eq:lemma_power} is true for $\left(\M{W}_{nm}\right)^{q - 1}$, i.e., 
	\begin{align}
	\left(\M{W}_{nm}\right)^{q-1} = \sum_{k = 0}^{d-1} \omega^{\left(\left(q-1\right) k + \frac{\left(q-1\right)\left(q-2\right)}{2} m \right)n}\ket{k}\bra{\modA{k}{\left(q-1\right) m}{d}},
	\label{eq:def:Weyl_l_1}
	\end{align}
	then, 
	\begin{align}
	\left(\M{W}_{nm}\right)^{q-1}\M{W}_{nm}  &= \sum_{k = 0}^{d-1} \omega^{\left(\left(q-1\right) k + \frac{\left(q-1\right)\left(q-2\right)}{2} m \right)n}\ket{k}\bra{\modA{k}{ \left(q-1\right) m}{d}}\sum_{j = 0}^{d - 1} \omega^{jn}\ket{j}\bra{\modA{j}{m}{d} },\\
	&= \sum_{k = 0}^{d-1} \omega^{\left(\left(q-1\right) k + \frac{\left(q-1\right)\left(q-2\right)}{2} m + k + \left(q - 1\right)m \right)n}\ket{k}\bra{\modA{k}{\left(q-1\right) m + m}{d}},\\
	&= \sum_{k = 0}^{d-1} \omega^{\left(q k +  \frac{q\left(q - 1\right)}{2}  m \right)n}\ket{k}\bra{\modA{k}{q m}{d}},
	\end{align}
	which recovers the expressions \eqref{eq:def:Weyl}, \eqref{eq:def:Weyl_2}, and \eqref{eq:def:Weyl_l_1} for $q = 1, 2,$ and $q - 1$, respectively, and proves the statement of our lemma.
\end{proof}

In the following Lemma we show that $\ell$ is always a divisor of $d$.

\begin{lemma}
	For a $d$-dimensional Weyl operator $\M{W}_{nm}$, the minimum $\ell\in \left\{1,2,\cdots , d \right\}$ such that $\left(\M{W}_{nm}\right)^\ell = pI$ is always a divisor of $d$.
\end{lemma}
\begin{proof}
	Assume that $\ell$ is not a divisor of $d$, then
	\begin{align}
	d = e\ell + r,
	\end{align}
	for some $e$ and $0< r < \ell $. Then,
	\begin{align}
	\left(\M{W}_{nm}\right)^d &= \left(\M{W}_{nm}\right)^{\left( e\ell + r\right)},\\
	&= \left(\M{W}_{nm}\right)^{e\ell}\left(\M{W}_{nm}\right)^r,\\
	&= p^e I \left(\M{W}_{nm}\right)^r,
	\end{align}
	which is not proportional to identity since $\ell$ is the least power such that $\left(\M{W}_{nm}\right)^\ell$ is proportional to identity, and $0 < r < \ell$. It is easy to see from Lemma~\ref{lemma:l_power} that $\left(\M{W}_{nm}\right)^d$ is always proportional to identity. Therefore, our starting assumption that $\ell$ is not a divisor of $d$ results into a contradiction and is wrong. Hence, the minimum $\ell$ such that  $\left(\M{W}_{nm}\right)^\ell = pI$ is always a divisor of $d$.
\end{proof}

Now writing the condition \eqref{eq:ell_power} explicitly in the Dirac's notation, i.e.,
\begin{align}
\sum_{k = 0}^{d-1} \omega^{\left(\ell k + \frac{\ell \left( \ell - 1\right)}{2} m \right)n}\ket{k}\bra{\modA{k}{\ell m}{d}} = \sum_{k = 0}^{d-1} p \ket{k}\bra{k},
\end{align}
we obtain
\begin{align}
p = \omega^{\left(\ell k + \frac{\ell \left( \ell - 1\right)}{2} m \right)n},
\label{eq:def_p}
\end{align}
and the following conditions on $\ell$
\begin{align}
\ell m &= a d \quad \because \modA{k}{\ell m}{d} = k \label{eq:con_m},\\
\ell n &= b d \quad \because \omega^{\left(\ell k + \frac{\ell \left( \ell - 1\right)}{2} m \right)n} = p \; \forall k,\label{eq:con_n}
\end{align}
where $0\leq a \leq m$ and $0\leq b\leq n$ are some positive integers. Note that since $\ell$ is the smallest number such that $\left( \M{W}_{nm}\right)^\ell$ is proportional to identity, therefore $\ell$ is the smallest number satisfying conditions \eqref{eq:con_m} and \eqref{eq:con_n}. It is easy to see from the minimality of $\ell$ that $a,b$, and $\ell$ are relatively prime (they do not have a common divisor greater than 1).

Using \eqref{eq:con_m} and \eqref{eq:con_n}, we can simplify \eqref{eq:def_p} as 
\begin{align}
p &= \omega^{\frac{\ell \left( \ell - 1\right)}{2} m n}\\
&= \begin{cases}
\omega^{-\ell n m /2},\quad &\text{ when $\ell$ is even}\\
\omega^{u\ell n m} = 1, \quad &\text{ when $\ell$ is odd},
\end{cases}
\label{eq:p_simplified}
\end{align}
where $u = \left( \ell - 1\right)/2$ is an integer.

Therefore, the eigenvalues of $\M{W}_{nm}$ are the $\ell$th roots of $p$, i.e., they are $\ell$ distinct numbers that satisfy
\begin{align}
p - x^\ell = 0,
\label{eq:l_equation}
\end{align}
where $p$ is given by \eqref{eq:p_simplified} and $\ell$ satisfies the conditions \eqref{eq:con_m} and \eqref{eq:con_n}. In the next section we show that \eqref{eq:eig_1} generates $\ell$ distinct numbers that satisfy \eqref{eq:l_equation}.

\subsubsection*{Generating $\ell$th Roots of $p$ by \eqref{eq:eig_1}}

As stated in the main text (given by \eqref{eq:eig_1} here) that the eigenvalues of any $\pmb{W}_{nm}$ on a $d$-dimensional Hilbert space are given by
\begin{align}
\lambda = \omega^{nm\frac{\left(d-1 \right)}{2} + \left( mk - nj \right)},
\label{eq:eig1}
\end{align}
for $j,k = 0,\cdots , d-1$; where we have substituted the variable $s$ with its expression for the clarity in the upcoming calculations. The validity of this expression can be established by showing that it generates all $\ell$th roots of $p$. We show this in two steps, i) every value generated by \eqref{eq:eig1} satisfies \eqref{eq:l_equation}, i.e., \eqref{eq:eig1} generates $\ell$th roots of $p$,  and ii) it generates exactly $\ell$ unique values, i.e., it generates \emph{all} $\ell$th roots of $p$.

Note that
\begin{align}
\lambda^\ell &= \omega^{\ell nm\frac{\left(d-1 \right)}{2} + \ell\left( mk - nj \right)},\\
&= \omega^{\ell nm\frac{\left(d-1 \right)}{2}},
\end{align}
where the last equality follows from conditions \eqref{eq:con_m} and \eqref{eq:con_n} and holds for every integer $k$ and $j$. 

Continuing with the last expression while considering the case when $\ell$ is even, we have
\begin{align}
\omega^{\ell nm\frac{\left(d-1 \right)}{2}} &= \omega^{\frac{\ell nmd}{2}} \omega^{-\frac{\ell n m}{2}},\\
&= \omega^{-\frac{\ell n m}{2}},\\
& = p,
\end{align}
where $\omega^{\frac{\ell nmd}{2}} = 1$, because $\ell /2$ is an integer and the exponent term is an integer multiple of $d$. 

The case when $\ell$ is odd can be further divided into the following two cases
\begin{enumerate}
	\item[i)] When $d$ is odd, $\left(d - 1\right)/2 = v$ is an integer, and using \eqref{eq:con_n}, we write
	\begin{align}
	\omega^{\ell nm\frac{\left(d-1 \right)}{2}} &= \omega^{bd m v},\\
	& = 1.
	\end{align}
	\item[ii)] When $d$ is even, the right hand sides of both \eqref{eq:con_m} and \eqref{eq:con_n} are even. Since $\ell$ is odd, therefore both $m$ and $n$ must be even for the equalities to hold. Hence, $m/2 = w$ is an integer, and we write
	\begin{align}
	\omega^{\ell nm\frac{\left(d-1 \right)}{2}} &= \omega^{bd w\left(d-1 \right)},\\
	& = 1.
	\end{align}
\end{enumerate}
Hence, we have shown that $\lambda^\ell = p$ for every integer $j$ and $k$, i.e., every value generated by \eqref{eq:eig} is indeed  an eigenvalue of $\M{W}_{nm}$. What is left to prove now is that it generates all the distinct eigenvalues of $\M{W}_{nm}$. This can be proven simply by showing that \eqref{eq:eig1} generates exactly $\ell$ distinct values.

Let $N$ be the total number of distinct values generated by \eqref{eq:eig}. Then, it is trivially true that $N\leq \ell$ since there are no more than $\ell$ distinct numbers that satisfy \eqref{eq:l_equation}. We need to show that \eqref{eq:eig1} generates at least $\ell$ distinct values. Since $nm\left( d - 1\right)/2$ in the exponent of \eqref{eq:eig1} is a constant for any given $n,m,$ and $d$, we only need to show that $\modS{mk}{nj}{d}$, for $0\leq j,k\leq d - 1$ generates at least $\ell$ distinct values. Let
\begin{align}
\modS{mj}{nk}{d} = h c ,
\label{eq:numbers}
\end{align}
where $c = d/\ell$, then showing the existence of $0\leq j,k \leq d - 1$ for every $0\leq h \leq \ell -  1$ shows that \eqref{eq:eig1} generates at least $\ell$ distinct values (because all $hc$ are distinct in the defined range). We substitute the values of $m$ and $n$ from \eqref{eq:con_m} and \eqref{eq:con_n} to obtain the equivalent condition
\begin{align}
\frac{ad}{\ell}j - \frac{bd}{\ell}k &= h \frac{d}{\ell} \\
\modS{aj}{bk}{\ell} &= h.
\label{eq:ab_condition}
\end{align}
Since $a,b$, and $\ell$ are relatively prime, the existence of $0\leq j,k \leq d - 1,$ in \eqref{eq:ab_condition} for every $0 \leq h \leq \ell -1$ is guaranteed by the B\'ezout's identity \cite{Tig:01:WS}. Therefore, \eqref{eq:eig1} generates at least $\ell$ distinct numbers. Since there are no more than $\ell$ distinct numbers that satisfy \eqref{eq:l_equation}, we deduce that \eqref{eq:eig1} generates exactly $\ell$ distinct values.

Hence we have shown that \eqref{eq:eig_1} is a generator of eigenvalues of $\M{W}_{nm}$, for $n,m = 0,1, \cdots , d-1$, and for any $d\geq 2$. This expression is particularly useful since it only depends on $n,m$, and $d$, and does not require any tedious calculations. Furthermore, it was previously known\cite{BHN:07:JPA} that the eigenvalues of Weyl operators are given by the $\ell$th roots of some complex number $p, \left| p \right| = 1$. However, specific properties of $\ell$ and $p$ were not known. On our way to proving the validity of \eqref{eq:eig_1}, we have derived certain properties of $\ell$, and an explicit expression for $p$. These new insights can be helpful in simply deriving some properties of Weyl operators. For example, it is straightforward to show that any Weyl operator (except for $\M{W}_{00}$) operating on a prime dimensional Hilbert space does not have any repeated eigenvalue (i.e., $\ell = d$) by using conditions \eqref{eq:con_m} and \eqref{eq:con_n} etc.



\section*{Acknowledgements}

This work was supported by the National Research Foundation of Korea (NRF) grant funded by the Korea government (MSIP) (No.~2016R1A2B2014462) and ICT R\&D program of MSIP/IITP [R0190-15-2030, Reliable crypto-system standards and core technology development for secure quantum key distribution network].

\section*{Author contributions statement}

J.R contributed the idea. J.R, J.K, and Y.J developed the theory. H.S improved the manuscript and supervised the research.  All the authors contributed in analyzing and discussing the results and improving the manuscript.

\textbf{Competing Interests}:  The authors declare that they have no competing interests.
\end{document}